\pdfoutput=1 

\documentclass[twoside,leqno]{article}
\usepackage{ltexpprt}

\usepackage[utf8]{inputenc}
\usepackage[english]{babel}
\usepackage{amsmath,amsfonts}
\usepackage{amssymb}
\usepackage[showonlyrefs]{mathtools}
\usepackage{moresize}
\usepackage[numbers,sectionbib,longnamesfirst]{natbib}
\usepackage[hyperindex,breaklinks,bookmarks,hidelinks]{hyperref}
\hypersetup{bookmarksdepth=2}
\usepackage[shortlabels]{enumitem}
\usepackage[toc]{appendix}
\usepackage{microtype}
\usepackage{xcolor}
\usepackage{mdframed}
\usepackage{cleveref}

\newmdenv[
  topline=false,
  bottomline=false,
  rightline=false,
  skipabove=\topsep,
  skipbelow=\topsep
]{lined}

\newcommand{\eps}{\varepsilon}

\date{}
\title{Testing Identity of Distributions under Kolmogorov Distance in Polylogarithmic Space}

\author{
	Christian Janos Lebeda\thanks{This work was carried out while this author was at the IT University of Copenhagen.}\\
	\texttt{christian-janos.lebeda@inria.fr}\\
	Inria, University of Montpellier
    \and
	Jakub Tětek\thanks{This author was supported by the VILLUM Foundation grants 54451 and 16582. This research was partially funded from the Ministry of Education and Science of Bulgaria (support for INSAIT, part of the Bulgarian National Roadmap for Research Infrastructure).
 Part of this work was carried out while this author was at BARC, University of Copenhagen.}\\
	\texttt{j.tetek@gmail.com}\\
	INSAIT, Sofia University ``St. Kliment Ohridski"
}

\begin{document}
\maketitle

\begin{abstract}
Suppose we have a sample from a distribution $D$ and we want to test whether $D = D^*$ for a fixed distribution $D^*$. Specifically, we want to reject with constant probability, if the distance of $D$ from $D^*$ is $\geq \eps$ in a given metric. In the case of continuous distributions, this has been studied thoroughly in the statistics literature. Namely, for the well-studied Kolmogorov metric a test is known that uses the optimal $O(1/\eps^2)$ samples.

However, this test naively uses also space $O(1/\eps^2)$, and previous work improved this to $O(1/\eps)$. In this paper, we show that much less space suffices -- we give an algorithm that uses space $O(\log^4 \eps^{-1})$ in the streaming setting while also using an asymptotically optimal number of samples. This is in contrast with the standard total variation distance on discrete distributions for which such space reduction is known to be impossible. Finally, we state 9 related open problems that we hope will spark interest in this and related problems.
\end{abstract}
%

\section{Introduction}
\label{sec:introduction}

Testing hypotheses about continuous distributions is a fundamental topic in statistics, with applications in a wide range of fields of science. However, in theoretical computer science -- and the area of distribution testing specifically -- the focus has been mostly on discrete distributions. As such, there are many algorithmic questions regarding testing continuous distributions, which remain unexplored.

One of the most fundamental problems in distribution testing is the so-called \emph{identity testing}. In this problem, we get a sample from a distribution $D$ and we want to know whether $D = D^*$ for a fixed distribution $D^*$. Namely, we want with probability $1-\delta$ to accept if $D = D^*$ and reject if $D$ has distance $\geq \eps$ from $D^*$ in some given metric.
This question has been studied extensively in statistics. Indeed, the famed Kolmogorov--Smirnov test (see \Cref{sec:preliminaries}) solves exactly this problem under the metric that is today called the Kolmogorov (K) distance, which is defined as the $L_\infty$ distance of the cumulative distribution functions. This test is known to be asymptotically optimal, using $\Theta(\frac{\log 1/\delta}{\eps^2})$ samples (see \Cref{sec:related-work} for details).

However, the Kolmogorov--Smirnov test has relatively high space complexity. If we assume that we receive the samples one by one and we can only loop over the data once,
it naively needs space 
$\Theta(\log(1/\delta)/\eps^2)$ 
for storing all samples.
This has been improved in previous work to space $\Theta(1/\eps)$ 
using two different approaches by \cite{cardoso2023online, lall15}.
We discuss their techniques further along with other related work in Section~\ref{sec:related-work}.
In this paper, we show that one can do much better. Namely, we show that one only needs polylogarithmic space for identity testing under Kolmogorov distance, while still achieving asymptotically optimal sample complexity. Specifically, we give a simple algorithm such that the following holds:

\smallskip \noindent
\textbf{Theorem~\ref{thm:main-theorem}, restated.}
\emph{
There exists an algorithm that (i) has sample complexity $O(\log(1/\delta)/\eps^2)$, (ii) has space complexity $O(\log^4 1/\eps)$, and (iii) it is correct for testing identity for any two distributions\footnote{It turns out we may assume the distributions are continuous without loss of generality. We do this throughout this paper, but our results also hold for non-continuous distributions. See the discussion in Section~\ref{sec:preliminaries}.} $D,D^*$ w.r.t.\ the Kolmogorov distance with probability at least $1-\delta$ (with probability $1-\delta$ it accepts if $D = D^*$ and rejects if $d_K(D,D^*) \geq \eps$).}
\smallskip \noindent

Our result is in contrast with identity testing under the standard total variation distance on discrete distributions. There, one cannot asymptotically decrease the space complexity by more than a logarithmic factor without asymptotically increasing the sample complexity \cite{diakonikolas2019communication}. On the one hand, without space constraints, sample complexity $\Theta(\sqrt{n}/\eps^2)$ is optimal \cite{valiant2017automatic} for $n$ being the support size of $D^*$. On the other hand, any algorithm that uses $s$ bits has to use $\Omega(\frac{n}{s \eps^2})$ samples \cite{diakonikolas2019communication}. For details on this and other related work, see \Cref{sec:related-work}. For even more background, we refer the reader to the recent survey on statistical inference under space constraints \cite{berg2023statistical}.

It should be noted that the space complexity of our algorithm can be further decreased at the cost of increasing the sample complexity. In fact, when explaining below the intuition behind our algorithm, we describe a simpler approach that leads to a lower space complexity but suboptimal sample complexity. This way, one can obtain an algorithm with constant space complexity and sample complexity $O(\frac{\mathrm{poly \, log (1/\eps)}}{\eps^2})$.

We hope that our approach will inspire further research into related topics. To this end, we state nine open problems in \Cref{sec:open_problems}.

\subsection{Our techniques}
We now sketch our approach. We describe here a less technical version that adds a polylogarithmic factor to the sample complexity. In the actual proof, we take care to ensure we do not asymptotically increase the sample complexity.

The high-level idea is that depending on $D^*$, we divide the real line into (non-disjoint) buckets bases on dyadic intervals. We then prove that if $d_K(D,D^*) \geq \eps$, then there has to be at least one of these buckets which has significantly different probabilities under the two distributions. We then estimate the probability of each bucket under the distribution $D$ and check whether it significantly differs from the probability under $D^*$, in which case we reject. If there is no such bucket, we accept.

The difficulty with this approach is that there will be more buckets than the amount of space we can use. Specifically, there are $\Theta(1/\varepsilon)$ buckets and we want an algorithm that uses only polylogarithmic space. This means that we cannot estimate the buckets' probabilities all at once using the same sample. Instead, we do it sequentially (first estimating the probability of some buckets, then of other buckets, and so on). The crux then lies in how to do this without increasing the sample complexity, since with this approach, we have to use fresh samples for estimating different buckets.

We get around this issue using the following idea: The lower probability a bucket has, the more such buckets there will be, but the fewer samples we need to estimate its probability to within a fixed additive error. We now describe this idea in more detail.


Suppose $d_K(D,D^*) \geq 10\eps$. We divide the real line into buckets each having probability $2^{-j}$ under $D^*$ for $j = 1, \cdots, \lceil \lg 1/\eps \rceil$ as follows. For each $j$, we let $B_{1,j}$ be a bucket starting with $-\infty$ that has probability $2^{-j}$, then $B_{2,j}$ is an interval with probability $2^{-j}$ whose left endpoint is the right endpoint of $B_{1,j}$, and so on. We claim that there has to be a bucket $B_{i,j}$ such that $|P[D \in B_{i,j}] - P[D^* \in B_{i,j}]| \geq \frac{\eps}{\lg 1/\eps}$.

To this end, we consider a point $y$ such that $|P[D \leq y] - P[D^* \leq y]| \geq 10\eps$ (note that such $y$ exists by the assumption that the K distance is $\geq 10\eps$).
We then argue that we can approximately write the interval $(-\infty, y]$ as a disjoint union of at most $\lceil \lg 1/\eps \rceil$ of the buckets $B_{i,j}$.  
We look at the binary expansion of the value of $P[D^* \leq y]$ and we consider the first $\lceil \lg 1/\eps \rceil$ bits in order. If the $j$-th bit is set to $1$, we add a bucket that has probability $2^{-j}$ in $D^*$ to the union, choosing the
bucket whose left endpoint is equal to the right endpoint of the bucket added last. This way, we obtain a set of buckets with size $\leq \lceil \lg 1/\eps \rceil$ that approximately covers the interval $(-\infty, y]$.
The remaining difference between the intervals is contained in a single bucket that has probability $\leq \varepsilon$ under $D^*$. 

Suppose for the sake of contradiction that the probability of each of the buckets under $D$ and $D^*$ differs by less than $\frac{\eps}{\lg 1/\eps}$. We then argue that by the triangle inequality, the probability of the interval $(-\infty, y]$ under the two distributions differ by less than $10\eps$. This is a contradiction with our definition of $y$. This means that there has to exist $i,j$ such that $|P[D \in B_{i,j}] - P[D^* \in B_{i,j}]| \geq \frac{\eps}{\lg 1/\eps}$. 

However, there are too many buckets for us to estimate all of their probabilities at once (in fact, there are $\Theta(1/\eps)$ different buckets, which we cannot afford to store at once due to space constraints). However, we claim that our job is easier for lower-probability buckets. We need to estimate the probabilities to within additive error $\frac{\eps}{\lg 1/\eps}$. By a standard calculation of variance and using the Chebyshev inequality, if the bucket's probability is $p$, we need $O(p \cdot \log^2 \eps^{-1} / \eps^2)$ samples (for constant failure probability, which we can amplify). This is very convenient for us: there are $2^j$ buckets on level $j$ each with probability $2^{-j}$. This means that we can afford to estimate their frequencies one by one using fresh samples each time for a total number of samples
\[
2^j \cdot O(2^{-j} \cdot \log^2 \eps^{-1} / \eps^{-2}) = O\left(\frac{\log^2 \eps^{-1}}{\eps^2}\right) \,.
\]
We do this at once separately on each of the $O(\lg \eps^{-1})$ levels, for a logarithmic space complexity and sample complexity $O\left(\frac{\log^2 \eps^{-1}}{\eps^2}\right)$ for constant failure probability. A more careful analysis can be used to shave off the logarithmic factors from the sample complexity at a polylogarithmic cost to the space complexity, thus obtaining at once polylogarithmic space complexity and asymptotically optimal sample complexity.

\section{Preliminaries}
\label{sec:preliminaries}

\paragraph{Notation and setup of the paper.}
We write $D$ for the unknown distribution from which we draw i.i.d. samples, and we write $D^*$ for the fixed known distribution. The support for both $D$ and $D^*$ can be any totally ordered set $U$.
For a distribution $D$, $D(S)$ denotes the probability of $S$ for $S$ being a measurable set.
We abuse notation and write $P[D \leq x]$ to denote $P_{X \sim D}[X \leq x]$ for a value $x$.
We denote the inverse of the CDF of $D^*$ as $q$, defined as $q_x = \sup\{y \mid P[D^* \leq y] \leq x\}$, where $q_0 = -\infty$ and $q_1 = \infty$ are special cases.
We use $[k]$ for $k \in \mathbb{N}$ to denote $\{1, \cdots, k\}$. We use $\mathbb{I}[\cdot]$ for a given condition to be the random event where the condition is satisfied.
Finally, we write $\lg x$ for the base-2 logarithm of any $x \in \mathbb{R^+}$.

\paragraph{Kolmogorov distance and Kolmogorov-Smirnov test.} 
Let us have two distributions $D_1,D_2$ over the same ordered set $U$. We define the Kolmogorov distance $d_K(D_1, D_2)$ as $\sup_{x \in U} |P_{X_1 \sim D_1}[X_1 \leq x] - P_{X_2 \sim D_2}[X_2 \leq x]|$. In other words, the Kolmogorov distance of two distributions is defined as the $L_\infty$ distance on their CDFs.
The Kolmogorov--Smirnov statistic is the distance between the emperical CDF of a sample and the reference CDF. 
Given a sorted sample $(X_1, X_2, \dots, X_n)$ of $n$ data points drawn from $D$ the standard Kolmogorov--Smirnov test computes the KS statistic as $\max_i(\max(i/n - q_{X_i}, q_{X_i} - (i-1)/n))$. 
The null hypothesis where $D = D^*$ is rejected if the KS statistic is sufficiently high based on $n$ and accepted failure probability.

\paragraph{Continuous vs. non-continuous distributions.} Throughout this paper, we assume that all distributions are continuous, that is, that any singleton event has probability zero. We claim that this is without loss of generality -- if a distribution $D$ on the universe $U$ is not continuous, we may replace it with a distribution $D' = D \otimes \text{Uniform}([0,1])$ on $U \times [0,1]$, with lexicographic order. If we have sampling access to $D$, we may easily simulate samples from $D'$ and at the same time, the K distance does not decrease (in fact, it stays the same) by performing this transformation.

\paragraph{Chernoff bounds.}
We use Chernoff bounds in the proofs for the correctness guarantees of our algorithm.
Specifically, we rely on Chernoff bounds to bound the value of a binomial distribution.
Here we list the specific bounds that are prerequisites for following the proofs. 
The forms of the bounds differ slightly from typical textbook presentations. 
We perform some simple calculations here to simplify the presentation of the proofs in \Cref{sec:one-sample-test}.

\begin{fact}
    \label{fact:chernoff_bound}
    Let $X$ be a random variable distributed as $Bin(n, p)$ and let $\mu = E[X] = np$. Then it holds that
    \begin{equation} \label{eq:chern-simpler-standard}
        P[X/n > p + \Delta] \leq e^{-\Delta^2n/(2p + \Delta)} \,.
    \end{equation}
    
    Furthermore, for any $\Delta < \mu$ it holds that
    \begin{equation} \label{eq:chern-less-than}
    P[X/n < p - \Delta] \leq e^{-\Delta^2n/(2p)} \,,
    \end{equation}
    
    and
    \begin{equation} \label{eq:chern-abs-distance}
    P[\vert X - \mu \vert > \Delta] \leq 2e^{-\Delta^2/(3\mu)} \,.
    \end{equation}
\end{fact}

\begin{proof}
    Throughout this proof we let $\delta = \Delta/p$.
    We find the first bound by
    \begin{align*}
        P[X/n > p + \Delta] = P[X > \mu + n\Delta] &= P[X > (1 + \delta)\mu] \leq \exp(-\delta^2\mu/(2 + \delta)) \\
        &= \exp(-(\Delta/p)^2 np/(2 + \Delta/ p)) =
        \exp(-\Delta^2n/(2p + \Delta)) \,,
    \end{align*}
    where the inequality is a common simpler variant of the standard Chernoff bound~(see \url{https://www.lkozma.net/inequalities_cheat_sheet/ineq.pdf}).
    
    The second bound follows from 
    \[
    P[X/n < p - \Delta] = P[X < \mu - n\Delta] = P[X < (1 - \delta)\mu] \leq e^{-\delta^2\mu/2} = e^{-\Delta^2n/(2p)} \,,
    \]
    where the inequality is a standard Chernoff-Hoeffding bound~\cite[Theorem~1.1]{dubhashi:panconesi:2009}.

    Finally, for the the third bound we have that  
    \[
        P[\vert X - \mu \vert > \Delta] \leq P[X < \mu - \Delta] + P[X > \mu + \Delta] = P[X < (1 - \Delta/\mu)\mu] + P[X > (1 + \Delta/\mu)\mu] \leq 2e^{-\Delta^2/(3\mu)}  \,,
    \]
    where the first inequality is a union bound and the second inequality follows from standard Chernoff-Hoeffding bounds~\cite[Theorem~1.1]{dubhashi:panconesi:2009}.
\end{proof}

\paragraph{Models of computation.} We consider the KS test in the one-pass streaming model. In this model, we receive the samples from $D$ one by one and access each sample only once.
Alternatively, we could consider a model where the algorithm has oracle access to samples from $D$, and our goal is to satisfy a certain sampling budget and space constraint.
The algorithm is equivalent in both models, but we present our work under the first model as it is more common in the computer science litterature.
We give our space complexity in the $w$-bit word RAM model as defined by~\cite{Hagerup98} for $w$ large enough so that a single word can represent a sample from the distribution or a counter of size $O(1/\eps^2 + \log \delta^{-1})$.
Each counter stored by our algorithm fits in a single word since all counts are at most this large.
We do not consider the space requirements for representing $q$  or $D^*$.
Our algorithm requires access to $q_{i2^{-j}}$ for each $i \in [2^j]$ where $j = \lceil \lg 1/\varepsilon \rceil + 2$.
We assume that these $\Theta(1/\varepsilon)$ entries of $q$ are either stored efficiently or can be computed when necessary from a compact representation of $D^*$. 

\section{Our algorithm}
\label{sec:one-sample-test}

In this section, we present our algorithm and prove our main theorem, Theorem~\ref{thm:main-theorem}. Throughout the proof, we work with failure probability $\delta = 1/10$. This can be amplified by standard probability amplification to general $\delta$ as we discuss in the proof of Theorem~\ref{thm:main-theorem}.

Recall that we have two distributions $D,D^*$, we may sample from $D$ and we know $D^*$.
We define a bucket as $B_{i,j} = (q_{(i-1)2^{-j}},q_{i 2^{-j}}]$, where we have $q_x = \sup \{y \mid P[D^* \leq y] \leq x\}$ with special cases $q_0 = -\infty$ and $q_1 = \infty$.


\begin{lined}
\textbf{Our algorithm.} We describe a subroutine for a fixed value of $j$. We then execute this subroutine for all $j \in [\lceil \lg 1/\eps \rceil + 2]$ in parallel (that is, we feed any sample 
as an input to all these subroutines).

\bigskip \noindent
For a given value $j$, we consider the first $\Theta( \log^3 1/\eps)$ buckets $B_{i,j}$ and we have a counter for each (we take fewer buckets if fewer remain). 
We then take $t = c \cdot \min(\frac{1}{\eps^2}, {\frac{(\lg 1/\eps + 3)^3}{2^j \eps^2}})$ samples where $c$ is a sufficiently high constant. 
Based on them, we estimate the probability under $D$ of each $B_{i,j}$ that we are storing. If the estimated probability differs from the probability under $D^*$ (that is, $2^{-j}$) by more than $\Delta = \max({\frac{\eps}{j^{2}},\frac{\eps}{\lg 1/\eps+3}})/20$, we reject. We then take the next $\Theta({\log^3 1/\eps})$ buckets and repeat, until we have gone through all $2^j$ buckets $B_{\cdot, j}$.

\bigskip \noindent
If none of the subroutines rejected, we accept.

\bigskip \noindent
Use probability amplification to amplify success probability to $1-\delta$: execute this algorithm sequentially $\Theta(\log 1/\delta)$ times and return the majority answer.
\end{lined}

\paragraph{Intuition behind parameter choice.} A curious reader might ask why we choose the values of $t$ and $\Delta$ the way we do. The intuition behind our algorithm that we described above corresponds to choosing $t = \Theta(\frac{(\lg 1/\eps)^3}{2^j \eps^2})$ and $\Delta = \Theta(\eps/\log(1/\eps))$. Namely, we want to check whether the probability of a bucket is within $\Theta(\eps/\log(1/\eps))$ of $2^{-j}$ and we want to be correct with probability $1-\Theta(\eps)$ so as to be able to union-bound over the $O(1/\eps)$ intervals. The number of samples this requires is equal to this choice of $t$, by Chernoff's inequality. The issue comes from the fact that for small values of $j$, the value of $t$ is greater than the total number of samples we have available, which is $\Theta(1/\eps^2)$.

A naive solution is to simply enforce that $t$ is never more than $1/\eps^2$. This is precisely what we choose to do. However, our estimates would then not be accurate enough. To this end, we relax the requirement for our estimates for smaller values of $j$. First, we have to settle for a higher failure probability. This is not a large issue -- we union bound the failure probability over different values of $j$, meaning that as long as the failure probability decreases sufficiently fast in $j$, the probabilities still add up to a constant.

Second, we also have to settle for lower accuracy -- for example, for $j=1$, we have $2^{-j} = 1/2$, in which case $1/\eps^2$ samples only gives us accuracy to within additive $\Theta(\eps)$, instead of the desired $\Theta(\eps/\log(1/\eps))$. In \Cref{lem:buckets-distance}, we will thus prove a stronger claim which allows us to have lower accuracy for smaller values of $j$. Namely, for small values of $j$, we test whether the probability differs from $2^{-j}$ by $\eps/j^2$. There is no reason why use specifically $j^2$ -- we could use any polynomial in $j$ or even an exponential with a sufficiently small base.

\paragraph{Analyzing our algorithm.} Before we are ready to prove correctness of our algorithm, we prove the following lemma, which states that if indeed $d_K(D,D^*) \geq \eps$, then there has to be a bucket $B_{i,j}$ whose probability differs significantly under the two distributions.

\begin{lemma}
\label{lem:buckets-distance}
Let us have two continuous distributions $D,D^*$ over a totally ordered set $U$. 
Suppose that $d_K(D,D^*) \geq \eps$. Then there exists $j \in [\lceil\lg 1/\eps\rceil +2]$ and $i \in [2^j]$ such that $|D^*(B_{i,j}) - D(B_{i,j})| \geq 2 \Delta =  \max\left(\frac{\eps}{j^{2}},\frac{\eps}{\lg 1/\eps+3}\right)/10$.
\end{lemma}
\begin{proof}
Let $x \in [0,1]$ be such that $\vert P[D \leq q_x] - P[D^* \leq q_x] \vert \geq \eps$. Note that such $x$ exists by the assumption $d_K(D,D^*) \geq \eps$. Our goal will be to approximately write the set $(-\infty, q_x]$ as a disjoint union of some of the buckets $B_{i,j}$. We then use this to show that if no such bucket exists, then it has to be the case that $|P[D \leq q_x] - P[D^* \leq q_x]| < \eps$, which would be a contradiction.

Let $b_1b_2\cdots$ be the binary expansion of $x$. We now define $x_{:\ell} = b_1\cdots b_\ell$ and we let $\tilde{x} = x_{:\lceil \lg 1/\eps \rceil + 2}$. Note that $x - \varepsilon/4 \leq \tilde{x} \leq x$. This means that $|P[D^* \leq q_{\tilde{x}}] - P[D^* \leq q_x]| = x - \tilde{x} \leq \eps/4$.

We use the binary encoding of $\tilde{x}$ to write the set $(-\infty, q_{\tilde{x}}]$ as a disjoint union of some of the buckets $B_{i,j}$. Namely, we define a set of indices $I$ such that $(-\infty, q_{\tilde{x}}] = \bigcup_{(i,j) \in I} B_{i,j}$ and such that this union is disjoint. 
We construct this set inductively. Namely, we define $I_0, \cdots, I_{\lceil \lg 1/\eps \rceil + 2}$ where for any $k$, $I_k$ satisfies that $(-\infty, q_{x_{:k}}] = \bigcup_{(i,j) \in I_{k}} B_{i,j}$. 
First, we set $I_{0} = \emptyset$. Suppose now that we have the set $I_{k}$; we inductively construct $I_{k+1}$. If $b_{k+1} = 0$ then $x_{:k+1} = x_{:k}$ and we may thus set $I_{k+1} = I_{k}$. If, on the other hand $b_{k+1}=1$, then it holds
\[
(-\infty, q_{x_{k+1}}] \setminus (-\infty, q_{x_{k}}] = (q_{x_{:k}}, q_{x_{:k+1}}] = B_{x_{:k}2^{k+1} + 1, k+1} \,, 
\]
meaning that we may set $I_{k+1} = I_k \cup \{(x_{:k}2^{k+1} + 1, k+1)\}$.
Finally, we define $I=I_{\lceil \lg 1/\eps \rceil + 2}$. 


Assume now for the sake of contradiction that $|D^*(B_{i,j}) - D(B_{i,j})| < 2\Delta = \max(\frac{\eps}{j^{2}},\frac{\eps}{\lg 1/\eps+3})/10$ for all $B_{i,j}$ where $(i,j) \in I$. We bound 
\begin{align*}
| P[D \leq q_x] - P[D^* \leq q_x]| =&\, |\sum_{i,j \in I} D(B_{i,j}) + P[D \in (q_{\tilde{x}}, q_x]] \, - \, \sum_{i,j \in I} D^*(B_{i,j}) - P[D^* \in (q_{\tilde{x}}, q_x]] | \\
\leq&\, \sum_{i,j \in I} |D(B_{i,j}) - D^*(B_{i,j})| + |P[D \in (q_{\tilde{x}}, q_x]] - P[D^* \in (q_{\tilde{x}}, q_x]]| \\
<&\, \sum_{j=1}^{\lceil \lg 1/\eps \rceil+2} \max\left(\frac{\eps}{j^{2}},\frac{\eps}{\lg 1/\eps+3}\right)/10 + |P[D \in (q_{\tilde{x}}, q_x]] - P[D^* \in (q_{\tilde{x}}, q_x]]| \\
<&\, \eps/2 + |P[D \in (q_{\tilde{x}}, q_x]] - P[D^* \in (q_{\tilde{x}}, q_x]]| \,,
\end{align*}
where we have used the calculation
\[
\sum_{j=1}^{\lceil \lg 1/\eps \rceil+2} \max \left(\frac{\eps}{j^{2}},\frac{\eps}{\lg 1/\eps+3}\right) < \sum_{j=1}^{\lceil \lg 1/\eps \rceil+2} \left(\frac{\eps}{j^{2}} + \frac{\eps}{\lg 1/\eps+3}\right) < 1.65 \cdot \eps + \eps < 5 \cdot \eps \,.
\]
Therefore, if the left-hand side is $\geq \eps$ as we assume, it has to be the case that $|P[D \in (q_{\tilde{x}}, q_x] - P[D^* \in (q_{\tilde{x}}, q_x]| > \eps/2$. As we have noted, it holds that $P[D^* \in (q_{\tilde{x}}, q_x]] \leq \eps/4$, meaning that it has to be the case that $P[D \in (q_{\tilde{x}}, q_x]] > \eps/2$. 
But this event is a subset of the event $D \in B_{\tilde{x}2^{j} + 1, j}$ 
for $j = \lceil \lg 1/\eps \rceil + 2$. The probability of outputting a value in this bucket under $D^*$ is $2^{-\lceil \lg 1/\eps \rceil + 2} \leq \eps/4$, so the difference of the probability under the two distributions is $> \eps/4$, which is $> 2\Delta$. This is a contradiction, concluding the proof.
\end{proof}

Next, we bound the failure probability of our algorithm. We separately consider the cases where $D$ is either equal to or $\varepsilon$-distinct from $D^*$. 
We bound the failure probability for a single execution of our algorithm by $1/10$ in each case. 
We then use probability amplification to achieve failure probability at most $\delta$.

\begin{lemma}
\label{lem:reject-distinct}
Suppose that $d_K(D,D^*) \geq \varepsilon$. 
Then our algorithm, without performing probability amplification, accepts with probability at most $1/10$.
\end{lemma}

\begin{proof}
By \Cref{lem:buckets-distance}, there exists a bucket $B_{i,j}$ such that $| D(B_{i,j}) - D^*(B_{i,j}) | = | D(B_{i,j}) - 2^{-j} | \geq 2 \Delta = \max(\frac{\varepsilon}{j^{2}}, \frac{\varepsilon}{\lg 1/\varepsilon + 3})/10$. The algorithm estimates the probability under $D$ of this bucket using $t = c \cdot \min (\frac{1}{\eps^2}, \frac{(\lg 1/\eps + 3)^3}{2^j \eps^2})$ samples; let $Z_{i,j}$ be this estimate. 
The algorithm can only accept if $|Z_{i,j} - 2^j| \leq \Delta$. 
We bound this probability.

We start with the case $D(B_{i,j}) \geq 2^{-j} + 2\Delta$. In this case, we prove that $P[Z_{i,j} < 2^{-j} + \Delta] \leq 1/10$. In general, the probability that a binomial distribution $Bin(n,p)$ is below a threshold for a fixed $n$ is a decreasing function of $p$. This means that the probability $P[Z_{i,j} < 2^{-j} + \Delta]$ is decreasing in $D(B_{i,j})$. As such we may assume that $D(B_{i,j}) = 2^{-j} + 2\Delta$. This allows us to use the Chernoff bound (see \Cref{eq:chern-less-than}) to bound:
\[
P[Z_{i,j} < 2^{-j} + \Delta] \leq \exp\left(-\frac{t\Delta^2}{2 (2^{-j} + 2\Delta)}\right) \leq \exp\left(-\frac{t\Delta^2}{4 \cdot 2^{-j}}\right)  \,,
\]
where the last inequality uses that $2^{-j} \geq 2^{-\lceil \lg 1/\eps \rceil - 2} \geq \eps/8$ whereas $\Delta \leq \eps/20$, meaning that $2^{-j} + 2 \Delta \leq 2 \cdot 2^{-j}$.
A similar argument can be used in the case $D(B_{i,j}) \leq 2^{-j} - 2\Delta$ to give the exact same upper bound. Similarly to the above, we may assume $D(B_{i,j}) = 2^{-j} - 2\Delta$ because we bound $P[Z_{i,j} > 2^{-j} - \Delta]$ which is increasing in $D(B_{i,j})$. 
In that case, we have $D(B_{i,j}) \geq \eps/8 - 2 \cdot \eps/20 = \eps/40$. This means that in this case, we have $\Delta/D(B_{i,j}) \leq 2$. This allows us to again use the Chernoff bound (\Cref{eq:chern-simpler-standard}).
\[
P[Z_{i,j} > 2^{-j} - \Delta] \leq \exp\left(-\frac{t\Delta^2}{2 \cdot D(B_{i,j}) + \Delta}\right) \leq \exp\left(-\frac{t\Delta^2}{4 (2^{-j} - 2\Delta)}\right) < \exp\left(-\frac{t\Delta^2}{4 \cdot 2^{-j}}\right) \,.
\]
It remains to argue that this is $\leq 1/10$.
%
We consider separately the two branches of the $\min$ in $t$.
First, we consider when $t = c/\eps^2$.
In this case, we substitute and bound
\[
\exp\left(-\frac{t\Delta^2}{4 \cdot 2^{-j}}\right) \leq \exp\left(-\frac{c}{\eps^2} \cdot \frac{\eps^2}{400 j^{4}} \cdot \frac{1}{4 \cdot2^{-j}}\right) = \exp\left(-\frac{c}{1600j^{4}2^{-j}}\right) < 0.1 \,,
\]
where the last inequality holds for a large enough value of $c$.
Similarly, if $t = c \cdot \frac{(\lg 1/\eps + 3)^3}{2^j \eps^2}$, we may bound
\[
\exp\left(-\frac{t\Delta^2}{4 \cdot 2^{-j}}\right) \leq \exp\left(-{c \cdot \frac{(\lg 1/\eps + 3)^3}{2^j \eps^2} \cdot \left(\frac{\eps}{20(\lg 1/\eps + 3)}\right)^2} \cdot \frac{1}{4 \cdot 2^{-j}}\right) \leq \exp\left(-\frac{c (\lg 1/\eps + 3)}{1600}\right) < 0.1 \,,
\]
where the last inequality again holds for $c$ large enough.
\end{proof}

\begin{lemma}
\label{lem:accept-equal}
Suppose that $D = D^*$. 
Then our algorithm, without performing probability amplification, rejects with probability at most $1/10$.
\end{lemma}

\begin{proof}
    The algorithm draws $t = c \cdot \min(1/\varepsilon^2, (\lg 1/\varepsilon + 3)^3/(2^j \varepsilon^2))$ samples to estimate the probability $D(B_{i,j})$, for $c$ being a sufficiently large constant that consider later.
    Let $X_{i, j}$ denote the random variable counting the number of these $t$ samples in the interval $B_{i,j}$.
    Our algorithm rejects only if for at least one bucket the event $|X_{i, j} - t/2^j| \geq t \Delta = t \cdot \max(\frac{\varepsilon}{j^{2}}, \frac{\varepsilon}{\lg 1/\varepsilon + 3})/20$ occurs.
    
    We define $I = \{i,j \mid j \in [\lceil \lg 1/\varepsilon \rceil + 2] \land i \in [2^j]\}$ and we define $I_1$ as the subset of $(i,j) \in I$ where $\frac{1}{\varepsilon^2} \leq \frac{(\lg 1/\varepsilon + 3)^3}{2^j \varepsilon^2}$ and $I_2$ the remaining elements of $I$. That is, $I_1$ corresponds to the indices where the first branch of the $\min$ in $t$ is smaller, whereas the second branch is smaller in $I_2$.
    Let $Y_{i,j} = \mathbb{I}[|X_{i, j} - t/2^j| \geq t \Delta]$ denote the event where we reject based on bucket $B_{i,j}$.
    We use a union bound for the probability of any $Y_{i,j}$ occurring.
    \[
        P\left[\bigcup_{i,j \in I} Y_{i,j} \right] \leq 
        \sum_{i,j \in I} P\left[ Y_{i,j} \right]
        = \sum_{i,j \in I_1} P\left[Y_{i,j} \right] + \sum_{i,j \in I_2} P\left[Y_{i,j} \right] \,.
    \]
    

    The random variable $X_{i,j}$ follows a binomial distribution with expected value of $\mu = t / 2^j$. 
    Using a standard Chernoff bound we have that 
    $P[\vert X_{i, j} - t / 2^j \vert \geq \delta] \leq 2e^{-\delta^2/(3\mu)}$ (see \Cref{eq:chern-abs-distance}).

    We first consider the case where $t = c / \varepsilon^2$, corresponding to the indices in $I_1$. We bound the probability of rejecting due to bucket $B_{i,j}$ as

    \begin{align*}
        P\left[ \vert X_{i,j} - t/2^j \vert \geq t \cdot \max\left(\frac{\varepsilon}{j^{2}}, \frac{\varepsilon}{\lg 1/\varepsilon + 3}\right)/20 \right] &\leq 
        P\left[ \vert X_{i,j} - t/2^j \vert \geq \frac{t \cdot \varepsilon}{20j^{2}} \right] \\
        \leq 2\exp\bigg(-\left(\frac{t \cdot \varepsilon}{20j^{2}}\right)^2/(3\mu)\bigg)
        &= 2\exp\bigg(-\left(\frac{c}{20j^{2}\varepsilon}\right)^2 \cdot \frac{\varepsilon^2 2^j}{3c}\bigg)
        \leq 2\exp\left(-\frac{200 \cdot 2^j}{j^{4}}\right) \,,
    \end{align*}
    where the last inequality holds for $c$ large enough.

    There are $2^j$ buckets for each value of $j \in [\lceil \lg 1/\varepsilon \rceil + 2]$. 
    It follows from a union bound that the probability of rejecting any of the buckets with $t = c/\varepsilon^2$ samples is upper bounded by

    \begin{align*}
         \sum_{i, j \in I_1} P\left[Y_{i,j}\right] \leq \sum_{j = 1}^{[\lceil \lg 1/\varepsilon \rceil + 2]} 2^j 2\exp\left(-\frac{200 \cdot 2^j}{j^{4}}\right) < \sum_{j = 1}^\infty 2^j 2\exp\left(-\frac{200 \cdot 2^j}{j^{4}}\right) < 0.05 \,.
    \end{align*}
    
    Next we consider the case where $t = \frac{c (\lg 1/\varepsilon + 3)^3}{2^j \varepsilon^2}$, corresponding to $I_2$. We bound the probability of rejecting based on bucket $B_{i,j}$ as

    \begin{align*}
        P\left[ \vert X_{i,j} - t/2^j \vert \geq t \cdot \max\left(\frac{\varepsilon}{j^{2}}, \frac{\varepsilon}{\lg 1/\varepsilon + 3}\right)/20 \right] &\leq 
        P\left[\vert X_{i,j} - t/2^j \vert \geq \frac{t \cdot \varepsilon}{20(\lg 1/\varepsilon + 3)}\right] \\
        \leq 2e^{-(\frac{t \cdot \varepsilon}{20(\lg 1/\varepsilon + 3)})^2/(3\mu)} 
        = 2e^{-(\frac{\varepsilon}{20(\lg 1/\varepsilon + 3)})^2 2^j t/3} 
        &= 2e^{-\frac{c \, (\lg 1/\varepsilon + 3)^3}{20^2 \cdot 3 \cdot (\lg 1/\varepsilon + 3)^2}} 
        \leq 2e^{-2(\lg 1/\varepsilon + 3)}
        \leq \varepsilon/e^6 \,,
    \end{align*}
    where the inequality second to last holds for $c$ large enough.

    The total number of buckets is bounded by $\sum_{j = 1}^{[\lceil \lg 1/\varepsilon \rceil + 2]} 2^j < 2^{\lg 1/\varepsilon + 4} = 16/\varepsilon$.
    As such, using a union bound we have that
    \[
    \sum_{i,j \in I_2} P\left[Y_{i,j} \right] < 16/\varepsilon \cdot \varepsilon/e^6 < 0.05 \,.
    \]       
\end{proof}
    
We are now ready to state the theorem, which summarizes the properties of our algorithm.

\begin{theorem}
\label{thm:main-theorem}There exists an algorithm that (i) has sample complexity $O(\log (1/\delta)/\eps^2)$, (ii) has space complexity $O(\log^4 1/\eps)$, and (iii) it is correct for testing identity for any two distributions $D,D^*$ w.r.t.\ the Kolmogorov distance with probability at least $1-\delta$ (with probability $1-\delta$ it accepts if $D = D^*$ and rejects if $d_K(D,D^*) \geq \eps$).
\end{theorem}
\begin{proof}
The space complexity is clearly as claimed: each subroutine uses space $O(\log^3 1/\eps)$ and we execute $\lceil\lg 1/\eps\rceil + 2$ subroutines at once. The probability amplification does not increase space complexity as it only needs one additional counter for storing how many of the executions accepted.
Our algorithm, without performing probability amplification, is correct with probability $9/10$ by Lemma~\ref{lem:reject-distinct} and~\ref{lem:accept-equal}. This is amplified to $1-\delta$.

It remains to argue the sample complexity. We argue that without probability amplification, it is $O(1/\eps^2)$. The claim then follows since probability amplification adds a factor $\log 1/\delta$, implying the sample complexity is as claimed. We are executing the subroutines in parallel, meaning that it suffices to argue that each subroutine uses $O(1/\eps^2)$ samples. The number of buckets for a fixed value of $j$ is $2^j$ and they are processed in batches of size $O(\log^3 1/\eps)$. This means that the complexity of a subroutine is
\[
\Bigl\lceil\frac{2^j}{\log^3 1/\eps} \Bigr\rceil \cdot c \cdot \min\left(\frac{1}{\eps^2}, \frac{(\lg 1/\eps + 3)^3}{2^j \eps^2}\right) = O\left( \frac{1}{\eps^2} \right) \,,
\]
where the equality can be seen by a simple calculation if we consider separately the cases $\frac{2^j}{\log^3 1/\eps} \leq 1$ and $\frac{2^j}{\log^3 1/\eps} > 1$. 
%
%
%
\end{proof}

\section{Related work}
\label{sec:related-work}

Apart from the papers referenced below, we refer the interested reader to the recent survey on statistical inference under memory constraints \cite{berg2023statistical}.

\paragraph{Statistics literature.} In the statistical literature, the problem of identity testing has been studied since the 1930s when Kolmogorov and Smirnov developed their test~\cite{smirnov1939estimation,kolmogorov1933sulla}. It is easy to see that this test is optimal up to constant factors: On one hand, the sample complexity of $O(\log (1/\delta)/\eps^2)$ immediately follows from the Dvoretzky–Kiefer–Wolfowitz inequality \cite{dvoretzky1956asymptotic}. On the other hand, it is well-known that one needs $\Omega(\log(1/\delta)/\eps^2)$ samples to distinguish between Bernoulli distributions with parameter $1/2$ and $1/2 + \eps$.\footnote{Specifically, the total variation distance of $Bin(m,1/2), Bin(m,1/2+\eps)$ is $TV = O(\eps \sqrt{m})$ \cite[formula 15]{roos2001binomial}, meaning that one has to set $m = \Omega(1/\eps^2)$ to make the distance constant. (Note that the success probability is $\leq 1-TV$.) A more detailed computation shows optimality also in terms of $\delta$.}
But the Kolmogorov--Smirnov test can be used to distinguish between these two distributions\footnote{We can add a value drawn from e.g. $\text{Uniform}([0,1/2])$ to the Bernoulli samples to obtain samples following continuous distributions. The distributions are $\varepsilon$-distinct as the probabilities of outputting a value $\leq 1/2$ clearly differ by $\varepsilon$.},
so the sample complexity of the Kolmogorov--Smirnov test is optimal.

Many generalizations and other tests have been developed in the literature. This includes the Cramér–von Mises test~\cite{cramer1928composition,vonmises1928} or the Anderson-Darling test \cite{anderson1954test} which both put more weight on the tails of the distribution or the multivariate version of the Kolmogorov--Smirnov test \cite{justel1997multivariate}. We refer the reader wanting to know more about this area to the chapter on goodness of fit in the book \cite{lehmann1986testing}.

\paragraph{Applied work on small-memory Kolmogorov--Smirnov test.} It has been studied how to estimate the Kolmogorov--Smirnov statistic in streams in small amount of memory. 
The estimate can then be used to accept or reject the hypothesis similar to the classical KS test.
This approach differs from our technique in that we directly reject or accept the hypothesis without computing an estimate of the KS statistic.
We discuss here three different techniques for estimating the KS statistic and argue why each requires $\Omega(1/\varepsilon)$ space.

An estimator with error at most $3\varepsilon$ based on approximate quantiles was presented by~\cite{lall15}.
Their technique relies on any sketch data structure that returns a data point from a stream with length $n$ with rank error at most $\varepsilon n$.
They extract from this sketch a subset of the stream of size at least $1/(2\varepsilon)$.
They use this subset to estimate the KS statistic with error at most $2\varepsilon$.
As such they must use space $\Omega(1/\varepsilon)$ even ignoring any auxiliary space required to maintain such a sketch.

The problem has also been studied in the applied literature by~\cite{cardoso2023online} in a setting where they want to efficiently estimate the KS statistic at any point in the stream.
They store $m$ non-overlapping buckets each with probability $1/m$ under the reference distribution $D^*$.
Since they are repeatably estimating the KS statistic they use a greedy heuristic to update the estimate in time $O(\log(m))$ after observing each data point.
Their data structure can be easily used to estimate the KS statistic with error $\leq 1/m$ in time $O(m)$.
As such they need $\Theta(1/\varepsilon)$ buckets to estimate the KS statistic with error $\varepsilon$.
We also use $\Theta(1/\varepsilon)$ buckets based on the CDF on $D^*$.
However, we only need to store a small subset of these buckets at any point since our technique relies on estimating the probability for individual buckets under $D$ rather than the full CDF.

Lastly, an estimator based on chunking and averaging was presented by~\cite{nguyen2017streamsuitablekolmogorovsmirnovtypetestbig}.
They split the stream into chunks of length $J$.
They compute the KS statistic on each chunk and store a running average.
They show that the running average approaches the expected value in the limit.
The value of $J$ should therefore be chosen sufficiently high to ensure that the expected value of the test is close to the actual Kolmogorov distance.
We have that $J = \Omega(1/\varepsilon)$ since changing a single data point in a sample of length $J$ can change the KS statistic by $1/J$.

\paragraph{Discrete distributions.} We now discuss identity testing for discrete distributions (under the standard total variation distance). This problem has also been studied in statistics, but the tests have suboptimal dependency on the support size, see the review by \citet{balakrishnan2018hypothesis} for a detailed historical background. The optimal complexity in terms of the support size was given by \cite{batu2001testing} and this was later improved to also be optimal in terms of $\eps$ by \cite{valiant2017automatic}. 
What is interesting to us is that this problem has also been studied under various constraints such as privacy constraints 
 (see \cite{acharya2021inference} and references within) or communication constraints \cite{diakonikolas2019communication,acharya2021interactive}. However, to us the most interesting is the study of identity testing under space constraints.

The reduction in \cite{goldreich2016uniform} implies that any result for the special case of uniformity testing also extends to identity testing. \citet{diakonikolas2019communication} shows an algorithm with space $s$ and sample complexity $O(\frac{n \log n}{s \eps^4})$ and they show a lower-bound which is near-matching in terms of the dependency on $n$ and $s$. They also show a lower-bound which matches up to a constant factor for $s$ not being too small. The upper-bound was later simplified by \citet{canonne2022uniformity}. In the setting where the number of samples available is unlimited, \citet{berg2022memory} study the necessary state of a finite state machine that tests uniformity.

\section{Open problems} \label{sec:open_problems}
Many interesting problems remain open in bounded-memory algorithms for distribution testing of non-discrete distributions. We now describe some of them.

\paragraph{Testing identity with sampling access to two distributions.} Suppose that instead of knowing one of the two distributions in question, we only have sampling access to both. This corresponds to the two-sample KS test in the statistics literature. Is it possible to solve this problem in polylogarithmic space as well?

\paragraph{Improving the number of samples used.} Our result ignored the constant in the sample complexity. Suppose that the KS test uses $S_{KS}$ samples for a given level of $\eps$. Is there an algorithm that uses $(1+\gamma)S_{KS}$ samples and uses memory $O(\log ^c \eps^{-1}/\gamma^c)$ for some constant $c$? If not, is there at least for each $\gamma>0$ an algorithm with space complexity $O(\log^c \eps^{-1})$.

\paragraph{Approximating the KS statistic in $o(1/\eps)$ space.} It would be of great practical interest if a streaming algorithm existed for approximating the KS statistic is space $o(1/\eps)$ (space $O(1/\eps)$ can be achieved using buckets as discussed in Section~\ref{sec:related-work}).
Our algorithm can be modified to give an approximation up to a factor of $O(\log \eps^{-1})$. Can better approximation be achieved in polylogarithmic space?

\paragraph{Multi-dimensional distributions.} Generalizations of the K distance exist that are defined for multi-dimensional distributions. 
Are there efficient streaming algorithms for identity testing also in this case?

\paragraph{Testing relationships between multiple distributions.}
Suppose we have sampling access to a distribution on $\mathbb{R}^2$. Give a streaming algorithm for testing whether these two coordinates are independent. 
Does the second coordinate stochastically dominate the first? Are the coordinates positively associated? 

\paragraph{Tests that are more efficient at the tails.} Design a space-efficient algorithm for identity testing with respect to some distance that allows for distinguishing smaller differences on the tails. In the statistics literature, 
the Anderson-Darling test~\cite{anderson1954test} or the Cramér-von Mises test~\cite{cramer1928composition, vonmises1928} solves this problem. 
Is it possible to get a more space-efficient algorithm that approximately implements either of these two tests?

\paragraph{Improve the space complexity.} The exponent in our algorithm is higher than one would perhaps hope. Is this inherent to the problem or can this be improved? Lower bounds would be especially interesting. It should be noted that in the lower bound, one will have to somehow use the fact that the algorithm has optimal sample complexity -- as we noted in the introduction, it is possible to achieve constant space complexity at the cost of slightly higher sample complexity.

\paragraph{Extend the results to tolerant distribution testing.} For identity testing, in standard distribution testing, we ought to accept if the distributions are equal and reject if their Kolmogorov distance is $\geq \eps$. In tolerant distribution testing, one ought to accept if the distance is $< \eps_1$ and reject if it is $\geq \eps_2$. Many results are significantly more difficult in this setting than with standard distribution testing. Is this also the case for identity testing w.r.t.\ the Kolmogorov distance in the streaming setting?

\paragraph{Distributed setting.} Suppose we are in the following setting, instead of the streaming setting. We have multiple machines with each storing part of the data. We want to test identity or any of the problems from the above open problems. How do we do this while communicating as few bits as possible between the machines?
%

\section*{Acknowledgments}

Lebeda carried out parts of this work at the IT University of Copenhagen.
The authors would like to thank Bernhard Haeupler for helping improve the presentation of this paper.
We thank Gautam Kamath for helpful comments on terminology.
Finally, we thank anonymous reviewers for suggestions and comments that helped improve the paper.

\bibliographystyle{plainnat}
\bibliography{literature}
\end{document}